\documentclass[11pt]{article}%
\pdfoutput=1 

\usepackage{amsmath,amsfonts,amssymb}
\usepackage{amsthm}
\usepackage{mathtools}
\usepackage{graphicx}
\usepackage[margin=1in]{geometry} 

\usepackage{caption}
\usepackage{subcaption}
\usepackage[linesnumbered,ruled,vlined]{algorithm2e}
\SetKwInput{KwInput}{Input}                
\SetKwInput{KwOutput}{Output}
\SetKwFor{RepTimes}{for}{do}{end}

\usepackage[pagebackref,linktocpage]{hyperref}
\hypersetup{
    colorlinks=true, 
    linkcolor=blue, 
    citecolor=blue, 
    urlcolor=blue 
}
\usepackage{microtype}
\usepackage{thm-restate}
\usepackage[capitalize,nameinlink]{cleveref}

\usepackage{tikz}
\usepackage{pgfplots}
\pgfplotsset{compat=1.13}
\usepackage{wrapfig}
\usepackage{enumerate}
\usepackage{braket}
\usepackage{qcircuit}
\usepackage{circledsteps} 
\usepackage{dsfont} 

\renewcommand{\backref}[1]{}

\renewcommand{\backrefalt}[4]{%
\ifcase #1 %
\or
[p.\ #2]%
\else
[pp.\ #2]%
\fi}

\newtheorem{theorem}{Theorem}

\newtheorem{corollary}[theorem]{Corollary}

\newtheorem{proposition}[theorem]{Proposition}

\newcommand{\Units}{\mathbb{T}}
\newcommand{\poly}{\mathrm{poly}}

\newcommand{\AND}{\normalfont\textsc{AND}} 
\newcommand{\OR}{\normalfont\textsc{OR}}
\newcommand{\NOT}{\normalfont\textsc{NOT}}

\newcommand{\CNOT}{\mathrm{CNOT}}
\newcommand{\ind}{\mathop{\mathds{1}}}

\newcommand{\range}[2]{{[#1\mathrel{:}#2]}}

\newcommand{\controlu}{*-=[][F]{\phantom{\bullet}}}
\newcommand{\multistate}[2]{*+{\hphantom{#2}} \POS[0,0].[#1,0] !C *{#2} \POS[0,0].[#1,0] \drop\frm{}}
\newcommand{\ghoststate}[1]{*+{\hphantom{#1}} }
\newcommand{\ccteq}[1]{\multistate{#1}{\cong}}
\newcommand{\ccteqg}{\ghoststate{\cong}}

\begin{document}

\title{Quantum Mass Production Theorems}
\author{William Kretschmer\thanks{University of Texas at Austin. \ Email:
\texttt{kretsch@cs.utexas.edu}. \ Supported by an NDSEG Fellowship.}}
\date{}
\maketitle

\begin{abstract}
We prove that for any $n$-qubit unitary transformation $U$ and for any $r = 2^{o(n / \log n)}$, there exists a quantum circuit to implement $U^{\otimes r}$ with at most $O(4^n)$ gates. This asymptotically equals the number of gates needed to implement just a \textit{single} copy of a worst-case $U$. We also establish analogous results for quantum states and diagonal unitary transformations. Our techniques are based on the work of Uhlig [Math. Notes 1974], who proved a similar mass production theorem for Boolean functions.
\end{abstract}

\section{Introduction}
\label{sec:intro}

If a computational task requires $c$ resources, then common sense dictates that repeating the same task $r$ times should require roughly $c \cdot r$ resources. In many settings, 
including query complexity \cite{JKS10-direct-sum} and communication complexity \cite{JRS03-direct-sum,BBCR10-direct-sum}, this intuition can be made rigorous: such results are known as \textit{direct sum theorems}. 
Closely related are \textit{direct product theorems},
which show that, with a fixed computational budget, the probability of successfully performing $r$ independent tasks decays in $r$. We recommend \cite[Chapter 1]{Dru12-thesis} for a good overview of the topic.

Nevertheless, direct sum and direct product theorems are not universal. Some computational settings exhibit a ``mass production'' phenomenon, in which the cost of performing the same task many times in parallel does \textit{not} scale linearly with the number of repetitions. A well-known example \cite{Pau76-disj,Dru12-thesis} is based on the circuit complexity of matrix-vector multiplication. For a matrix $M \in \{0,1\}^{n \times n}$, define $f_M: \{0,1\}^n \to \{0,1\}^n$ by $f_M(v) = Mv$, where addition and multiplication are taken mod $2$. Then a simple counting argument implies that for most $M$, the complexity of implementing $f_M$ via a Boolean circuit is at least $\Omega(n^2 / \log n)$, as measured by the number of $2$-bit $\AND$, $\OR$, and $\NOT$ gates. Yet, by observing that $f_M^n$ (i.e. $f_M$ repeated $n$ times) is simply a matrix-matrix multiplication, we find that the cost of implementing $f_M^n$ is only $O(n^\omega)$, where $\omega < 2.38$ is the exponent of matrix multiplication \cite{AVW21-mmul,DWZ22-mmul}---substantially less than the naive bound of $O(n^3)$.

One might be left with the impression that such mass production phenomena can only occur for extremely special functions, like matrix multiplication, that have a particular algebraic or combinatorial structure. Remarkably, this intuition fails dramatically in the setting of Boolean circuit complexity. A theorem of Uhlig \cite{Uhl74-russian,Uhl74-english,Uhl92-multiple} shows that for \textit{any} Boolean function $f: \{0,1\}^n \to \{0,1\}$ and for any $r = 2^{o(n / \log n)}$, there exists a Boolean circuit implementing $f^r$ with at most $O\left(\frac{2^n}{n}\right)$ gates. Asymptotically, this equals the number of gates needed to evaluate a worst-case $f$ on a \textit{single} input, by the well-known counting argument of Shannon \cite{Sha49-circuit}. In fact, Uhlig even showed that the leading constant in the big-$O$ does not increase with $r$, and hence arbitrary Boolean functions can be mass produced with essentially no overhead.

\subsection{This Work}
In this work, we consider the natural question of whether a similar mass production phenomenon holds for quantum circuit complexity. Our question is well-motivated by recent works demonstrating that for certain learning tasks, algorithms with access to many copies of a quantum state on a quantum memory can be exponentially more powerful than algorithms that have access only to single copies of the state at a time \cite{CCHL22-memory,HBC+22-experiments,Car22-ptm}. Indeed, these results suggest that optimizing the complexity of mass producing quantum states and processes could have valuable applications. We also view our question as interesting from a purely theoretical perspective, especially considering that Uhlig's theorem for classical functions has recently found complexity-theoretic applications in characterizing the minimum circuit size problem \cite{RS21-kt,Hir22-mcsp}.

For simplicity, we consider quantum circuit complexity in the setting of qubit quantum circuits, using the universal gate set of arbitrary single-qubit gates plus $\CNOT$ gates with all-to-all connectivity. We also allow ancilla qubits initialized to $\ket{0}$, so long as they are reset to $\ket{0}$ at the end of the computation. We measure circuit complexity in terms of the $\CNOT$ count. This measure is justified by the fact that multiple-qubit gates are more error-prone and expensive to implement than single-qubit gates, and also by the observation that the number of single-qubit gates is related to the $\CNOT$ count by at most a factor of $4$ in any irredundant circuit.

In analogy with Uhlig's theorem \cite{Uhl74-russian,Uhl74-english,Uhl92-multiple}, our main result establishes mass production theorems for both quantum states and unitary transformations.

\begin{restatable}{theorem}{thmstatemassproduce}
\label{thm:state_mass_prod}
Let $\ket{\psi}$ be an $n$-qubit quantum state, and let $r = 2^{o(n / \log n)}$. Then there exists a quantum circuit with at most $(1 + o(1))2^n$ $\CNOT$ gates to prepare $\ket{\psi}^{\otimes r}$.
\end{restatable}

\begin{restatable}{theorem}{thmunitarymassproduce}
\label{thm:unitary_mass_prod}
Let $U$ be an $n$-qubit unitary transformation, and let $r = 2^{o(n / \log n)}$. Then there exists a quantum circuit with at most $(5/2 + o(1))4^n$ $\CNOT$ gates to implement $U^{\otimes r}$.
\end{restatable}

Note that the factor $2^n$ (respectively, $4^n$), in \Cref{thm:state_mass_prod} (respectively, \Cref{thm:unitary_mass_prod}) is optimal, because it asymptotically equals the number of $\CNOT$ gates needed to prepare a \textit{single} copy of an arbitrary $n$-qubit state (respectively, to implement an arbitrary $n$-qubit unitary once),  up to a small multiplicative constant \cite{SBM06-synthesis}. Above, we made the leading constants explicit only to illustrate that they are not too large, and thus to demonstrate that these theorems have some hope of becoming practical. We leave a full optimization of these constants and the factors hidden in the $o(1)$ to future work.


\subsection{Proof Overview}
Our results build heavily on the simple proof of Uhlig's theorem given in \cite{Uhl92-multiple}, which we now briefly summarize. The proof proceeds by first showing that for an arbitrary $f: \{0,1\}^n \to \{0,1\}$, one can compute $2$ copies of $f$ using roughly $\frac{2^n}{n}$ gates---the same cost as is needed to compute a single copy of a worst-case $f$. Then, Uhlig shows that we can generalize to a larger number of repetitions $r$ by a straightforward recursive argument. So, we focus on the $r = 2$ case.

Fix a parameter $k$ do be chosen later, and define for each $0 \le i \le 2^k - 1$ the function $f_i: \{0,1\}^{n-k} \to \{0,1\}$ to be the restriction of $f$ obtained by fixing the first $k$ bits to be the binary representation of $i$. So, for example,
\[
f(\underbrace{0,0,\ldots,0}_{k \text{\rm \ times}},x_{k+1},\ldots,x_n) = f_0(x_{k+1},\ldots,x_n).
\]
Next, we define a set of functions $g_\ell: \{0,1\}^{n-k} \to \{0,1\}$ for each $0 \le \ell \le 2^k$ by:
\begin{itemize}
\item $g_0 = f_0$.
\item $g_\ell = f_{\ell - 1} \oplus f_\ell$ if $1 \le \ell \le 2^k - 1$.
\item $g_{2^k} = f_{2^k - 1}$.
\end{itemize}

Observe that
\begin{equation}
f_i = \bigoplus_{\ell = 0}^{i} g_\ell = \bigoplus_{\ell = i+1}^{2^k} g_\ell.
\end{equation}
Now, suppose that we have a pair of inputs $x,y \in \{0,1\}^n$ to $f$, and our goal is to evaluate $f(x)$ and $f(y)$ simultaneously. Let $i$ and $j$ denote the integers whose binary representations are the first $k$ bits of $x$ and $y$, respectively. Assume without loss of generality that $i \le j$. Uhlig's idea is to evaluate $f(x)$ using the decomposition $f_i = \bigoplus_{\ell = 0}^{i} g_\ell$ and $f(y)$ using $f_j = \bigoplus_{\ell = j+1}^{2^k} g_\ell$. The key observation is that in doing so, we only need to evaluate each $g_\ell$ at most once. The cost of computing $f(x)$ and $f(y)$ this way is dominated by computing the $g_\ell$s. So, the total size of the circuit is roughly
\[
\left(2^k + 1\right) \left(\frac{2^{n-k}}{n-k} \right),
\]
because there are $2^k + 1$ different $g_\ell$s, and each $g_\ell$ is a function on $n - k$ bits. For reasonable choices of $k$, this is asymptotically $(1 + o(1))\frac{2^n}{n}$, as desired.

Our main insight is that the same general approach generalizes straightforwardly from mass producing Boolean functions to mass producing diagonal unitary matrices, which we establish in \Cref{thm:diagonal_mass_produce}. In one sense, the only conceptual change between our proof and Uhlig's is that we work with the group of complex units under multiplication, rather than the group $\{0,1\}$ under XOR. Nevertheless, our proof requires some care, as we do not deal with diagonal matrices directly. Rather, we mass produce the direct sum of a diagonal unitary with its inverse. In other words, for an $n$-qubit diagonal unitary $U$, we find it more convenient to work with the diagonal unitary on $n+1$ qubits that applies $U$ when the last qubit is $\ket{0}$, and $U^\dagger$ when the last qubit is $\ket{1}$. The intuitive reason why we require this change is that the XOR function is its own inverse, whereas multiplication by a complex unit is generally not.

Finally, once we have established \Cref{thm:diagonal_mass_produce} for diagonal unitary transformations, we obtain the mass production theorems for quantum states and general unitary transformations by using well-known decompositions of states and unitaries into diagonal gates \cite{SBM06-synthesis}.


%
%
%
%
%
%

\section{Preliminaries}

\subsection{Basic Notation}
We denote by $\ind\{p\}$ the function that evaluates to $1$ if proposition $p$ is true, and $0$ otherwise. If $\alpha$ is a complex number, we let $\alpha^*$ denote its complex conjugate. We denote by $\Units \coloneqq \{a + bi : |a|^2 + |b|^2 = 1\}$ the set of complex units. For a function $f: \{0,1\}^n \to \Units$, denote by $\bar{f}: \{0,1\}^{n+1} \to \Units$ the function defined by $\bar{f}(x, c) = f(x)^{1 - 2c}$, so that $\bar{f}$ evaluates to $f$ when $c = 0$ and evaluates to $f^*$ when $c = 1$. We freely identify a function $f: \{0,1\}^n \to \Units$ with the corresponding diagonal unitary transformation $U$ that acts as $U\ket{x} = f(x)\ket{x}$ on basis states $x \in \{0,1\}^n$.

We use standard notation for quantum circuits, including $\CNOT$, Toffoli, and Fredkin gates. We also borrow a large amount of notation and terminology from \cite{SBM06-synthesis}, as we detail further below. We define the $x$-, $y$-, and $z$-axis rotations by:
\begin{align*}
R_x(\theta) &= \begin{pmatrix}
\cos(\theta / 2) & i \sin(\theta / 2)\\
i \sin(\theta/2) & \cos(\theta / 2)
\end{pmatrix},\\
R_y(\theta) &= \begin{pmatrix}
\cos(\theta / 2) & \sin(\theta / 2)\\
- \sin(\theta/2) & \cos(\theta / 2)
\end{pmatrix},\\
R_z(\theta) &= \begin{pmatrix}
e^{-i \theta / 2} & 0\\
0 & e^{i \theta / 2}
\end{pmatrix}.
\end{align*}

\subsection{Multiplexors}

A \textit{multiplexor} with $s$ select qubits and $d$ data qubits is a block-diagonal $(s+d)$-qubit unitary transformation that preserves every computational basis state $\ket{x}$ on the select qubits. For brevity, we call such a unitary an $(s,d)$-multiplexor. An $(s,1)$-multiplexor in which all of the diagonal blocks are $R_z$ on the data qubit may alternatively be called a \textit{multiplexed $R_z$} (analogously for $R_x$ and $R_y$). Collectively, multiplexed $R_x$, $R_y$, and $R_z$ are called \textit{multiplexed rotations}. Observe that an $(s,1)$-multiplexed $R_z$ is equivalent to a unitary implementing $\bar{f}$ for some $f: \{0,1\}^s \to \Units$.

We require the following basic fact about implementing multiplexed rotations:

\begin{proposition}[{\cite[Theorem 8]{SBM06-synthesis}}]
\label{prop:multiplexed_rotation_naive}
Let $U$ be an $(n,1)$-multiplexed rotation. Then there exists a quantum circuit with at most $2^n$ $\CNOT$ gates to implement $U$.
\end{proposition}

\subsection{Generic Gates}
As in \cite{SBM06-synthesis}, we use circuit diagrams containing \textit{generic gates}. An equivalence of two circuit diagrams containing generic gates means that for any assignment of parameters to the generic gates on one side, there exists an assignment of parameters to the gates on the other side that makes the two circuits compute the same operator. We use the following notation for generic gates:

\bigskip
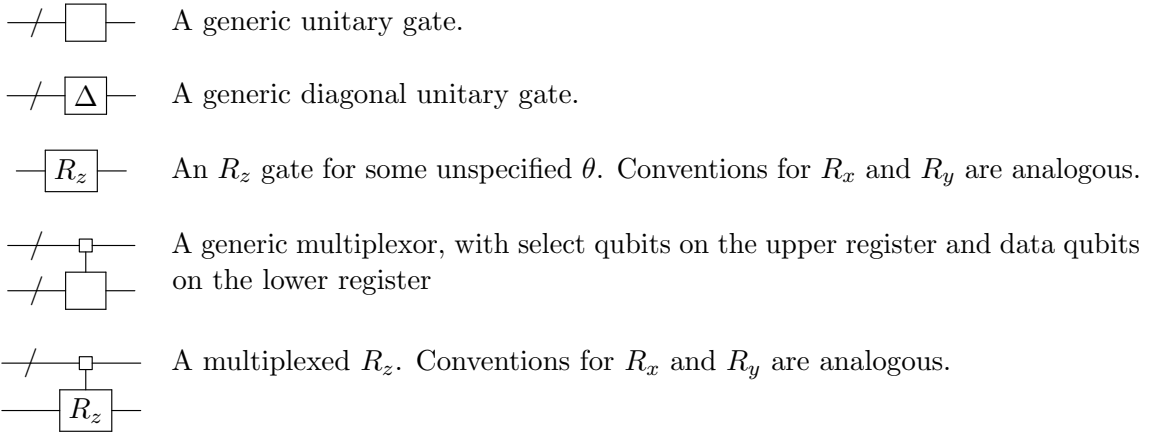
\begin{tabular}{cp{0.78\textwidth}}
$\Qcircuit @C=1em @R=.7em { & {/} \qw & \gate{\phantom{U}} & \qw }$ &
A generic unitary gate.\\
& \\
$\Qcircuit @C=1em @R=.7em { & {/} \qw & \gate{\Delta} & \qw }$ &
A generic diagonal unitary gate.\\
& \\
$\Qcircuit @C=1em @R=.7em { & \gate{R_z} & \qw }$ & An $R_z$ gate for some unspecified $\theta$. Conventions for
$R_x$ and $R_y$ are analogous.\\ & \\
$\Qcircuit @C=1em @R=.7em {
& {/} \qw & \controlu \qw & \qw\\
& {/} \qw & \gate{\phantom{U}} \qwx & \qw\\
}$ & A generic multiplexor, with select qubits on the upper register and data qubits on the lower register\\ & \\
$\Qcircuit @C=1em @R=.7em {
& {/} \qw & \controlu \qw & \qw\\
& \qw & \gate{R_z} \qwx & \qw\\
}$ & A multiplexed $R_z$. Conventions for
$R_x$ and $R_y$ are analogous.
\end{tabular}

\section{Diagonal Unitaries and Multiplexors}

We begin by generalizing the proof of Uhlig's theorem \cite{Uhl92-multiple} to diagonal unitary matrices (or, more precisely, multiplexed $R_z$ gates).

\begin{theorem}
\label{thm:diagonal_mass_produce}
Let $f: \{0,1\}^n \to \Units$ and let $r = 2^{o(n / \log n)}$. Then there exists a quantum circuit with at most $(1 + o(1))2^n$ $\CNOT$ gates to implement $\bar{f}^{\otimes r}$.
\end{theorem}

\begin{proof}
Without loss of generality, let $r = 2^t$ for some $t = o(n / \log n)$. Our proof proceeds by induction on $t$: for fixed $k$ (chosen later) and for every $n > k \cdot t$, we construct for each $f : \{0,1\}^n \to \Units$ a circuit $\mathcal{C}_{f,n,k,t}$ computing $\bar{f}^{\otimes 2^t}$. We proceed in order: first we construct $\mathcal{C}_{f,n,k,1}$ for every $n$ and $f$, then $\mathcal{C}_{f,n,k,2}$ for every $n$ and $f$, then $\mathcal{C}_{f,n,k,3
}$ for every $n$ and $f$, and so on. Ultimately, we show that there exists a universal constant $d$ such that the number of $\CNOT$ gates in $\mathcal{C}_{f,n,k,t}$, denoted $s_{n,k,t}$, satisfies the bound:
\begin{equation}
\label{eq:bound_s_nkt}
s_{n,k,t} \le \left(2^k + 1\right)^t \left(2^{n - tk} + 2^t dn \right).
\end{equation}

We begin by describing the construction of $\mathcal{C}_{f,n,k,1}$. For each $0 \le i \le 2^k - 1$, let $f_i: \{0,1\}^{n - k} \to \Units$ denote the restriction of $f$ obtained by fixing the first $k$ bits to the binary representation of $i$. For each $0 \le i \le 2^k$, define $g_i: \{0,1\}^{n - k} \to \Units$ by:
\begin{itemize}
\item $g_0 = f_0$.
\item $g_\ell = f_{\ell - 1}^* f_\ell$ if $1 \le \ell \le 2^k - 1$.
\item $g_{2^k} = f_{2^k - 1}^*$.
\end{itemize}

Observe that
\begin{equation}
f_i = \prod_{\ell = 0}^{i} g_\ell = \prod_{\ell = i+1}^{2^k} g_\ell^*.\label{eq:f_i_decomp}
\end{equation}
The key idea in the remainder of the proof is to evaluate $\bar{f}$ on a pair of inputs $(x, y)$ using the two decompositions in \eqref{eq:f_i_decomp}, one each for $x$ and $y$. Indeed, the following algorithm accomplishes this.

\begin{algorithm}[H]
\caption{Evaluate $\bar{f}^{\otimes 2}$}
\label{alg:eval_f_ifs}
\DontPrintSemicolon
\KwInput{$x, y \in \{0,1\}^n$, $c_x, c_y \in \{0,1\}$}
\KwOutput{$\bar{f}(x, c_x) \cdot \bar{f}(y, c_y)$}

$\alpha \coloneqq 1$

\uIf(\tcc*[f]{viewing $x,y$ as integers with highest order bits $x_1,y_1$}){$x \le y$}{
$m \coloneqq x$; $c_m \coloneqq c_x$ \tcc*[r]{set $m = \min\{x, y\}$, $M = \max\{x, y\}$}

$M \coloneqq y$; $c_M \coloneqq c_y$
}

\uElse{
$m \coloneqq y$; $c_m \coloneqq c_y$

$M \coloneqq x$; $c_M \coloneqq c_x$
}

\RepTimes{$0 \le \ell \le 2^k$}{
\uIf(\tcc*[f]{$x_\range{i}{j}$ denotes bits $i$ through $j$ of $x$}){$\ell \le m_\range{1}{k}$}{
Multiply $\alpha$ by $\bar{g}_\ell(m_\range{k+1}{n}, c_m)$
}
\uElseIf{$\ell > M_\range{1}{k}$}{
Multiply $\alpha$ by $\bar{g}_\ell(M_\range{k+1}{n}, 1 - c_{M})$ \tcc*[r]{note negation on $c_{M}$}
}
\uElse{Multiply $\alpha$ by $1$}
}

\Return{$\alpha$}
\end{algorithm}

Here, the $\ell \le m_\range{1}{k}$ clause corresponds to the multiplication $\prod_{\ell = 0}^{m_\range{1}{k}} g_\ell$, while the $\ell > M_\range{1}{k}$ clause corresponds to $\prod_{\ell = M_\range{1}{k}}^{2^n} g_\ell^*$. An equivalent reformulation of \Cref{alg:eval_f_ifs} is given below.

\begin{algorithm}[H]
\caption{Evaluate $\bar{f}^{\otimes 2}$}
\label{alg:eval_f_fixed}
\DontPrintSemicolon
\KwInput{$x, y \in \{0,1\}^n$, $c_x,c_y \in \{0,1\}$}
\KwOutput{$\bar{f}(x, c_x) \cdot \bar{f}(y, c_y)$}

$\alpha \coloneqq 1$

\uIf{$x \le y$}{
$m \coloneqq x$; $c_m \coloneqq c_x$

$M \coloneqq y$; $c_M \coloneqq c_y$
}

\uElse{
$m \coloneqq y$; $c_m \coloneqq c_y$

$M \coloneqq x$; $c_M \coloneqq c_x$
}

\RepTimes{$0 \le \ell \le 2^k$}{
$a \coloneqq \ind \{\ell \le m_\range{1}{k}\}$ \tcc*[r]{at most one of $a, b$ is nonzero}

$b \coloneqq \ind \{\ell > M_\range{1}{k}\}$

$z \coloneqq a \cdot m_\range{k+1}{n} \oplus b \cdot M_\range{k+1}{n}$

$c \coloneqq a \cdot c_m \oplus b \cdot (1 - c_M)$

Multiply $\alpha$ by $\bar{g}_\ell(z, c)$ 

Multiply $\alpha$ by $g_\ell^*(0^{n-k})^{(1 - a) \cdot (1 - b)}$ \tcc*[r]{undo added phase in case $a = b = 0$}
}

\Return{$\alpha$}
\end{algorithm}

\Cref{alg:eval_f_fixed} readily extends to a quantum circuit implementation. Define a pair of classical reversible circuits $\mathcal{A}_n$ and $\mathcal{B}_{n,k,\ell}$ whose input and output behavior are given in \Cref{fig:a_b}. Using $\mathcal{A}_n$ and $\mathcal{B}_{n,k,\ell}$, via the same strategy as \Cref{alg:eval_f_fixed}, we obtain the quantum circuit $\mathcal{C}_{f,n,k,1}$ defined in \Cref{fig:C_f_n_k_1} that implements $\bar{f}^{\otimes 2}$.

\begin{figure}[h]
\centering
\begin{subfigure}[c]{0.45\textwidth}
\[
\Qcircuit @C=1em @R=0.7em {
\lstick{x \in \{0,1\}^n}
& {/} \qw
& \multigate{2}{\mathcal{A}_n}
& \rstick{x} \qw
\\
\lstick{y \in \{0,1\}^n}
& {/} \qw
& \ghost{\mathcal{A}_n}
& \rstick{y} \qw
\\
\lstick{0}
& \qw
& \ghost{\mathcal{A}_n}
& \rstick{\ind\{x > y\}} \qw
\\
}
\]
\caption{}
\end{subfigure}
\qquad
\begin{subfigure}[c]{0.45\textwidth}
\[
\Qcircuit @C=1em @R=0.7em {
\lstick{0}
& \qw
& \multigate{3}{\mathcal{B}_{n,k,\ell}}
& \rstick{\ind\{\ell \le m_\range{1}{k}\}} \qw
\\
\lstick{0}
& \qw
& \ghost{\mathcal{B}_{n,k,\ell}}
& \rstick{\ind\{\ell > M_\range{1}{k}\}} \qw
\\
\lstick{m \in \{0,1\}^n}
& {/} \qw
& \ghost{\mathcal{B}_{n,k,\ell}}
& \rstick{m} \qw
\\
\lstick{M \in \{0,1\}^n}
& {/} \qw
& \ghost{\mathcal{B}_{n,k,\ell}}
& \rstick{M} \qw
\\
}
\]
\caption{}
\end{subfigure}
\caption{Inputs and outputs of reversible circuits $\mathcal{A}_n$ and $\mathcal{B}_{n,k,\ell}$.}
\label{fig:a_b}
\end{figure}
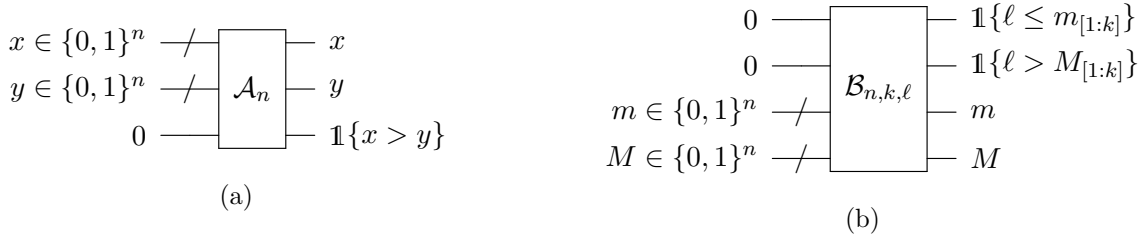

\begin{figure}
\small
\[
\Qcircuit @C=1em @R=1.2em @!R {
\lstick{\ket{0^{n-k}}}
& {/} \qw
& \qw
& \qw
& \qw
& \qw
& \qw
& \qw
& \qw
& \qw
& \targ
& \targ
& \qw
& \qw
& \ustick{_{z}} \qw
& \multigate{1}{\bar{g}_\ell}
& \qw
& \qw
& \targ
& \targ
& \qw
& \qw
& \qw
& \qw
& \qw
\\
\lstick{\ket{0}}
& \qw
& \qw
& \qw
& \qw
& \qw
& \qw
& \qw
& \qw
& \qw
& \qw
& \qw
& \targ
& \targ
& \ustick{_{c}} \qw
& \ghost{\bar{g}_\ell}
& \targ
& \targ
& \qw
& \qw
& \qw
& \qw
& \qw
& \qw
& \qw
\\
\lstick{\ket{0}}
& \qw
& \qw
& \qw
& \qw
& \qw
& \qw
& \qw
& \multigate{3}{\mathcal{B}_{n,k,\ell}}
& \ustick{_{a}} \qw
& \ctrl{-2}
& \qw
& \ctrl{-1}
& \qw
& \qw
& \multigate{1}{*}
& \qw
& \ctrl{-1}
& \qw
& \ctrl{-2}
& \multigate{3}{\mathcal{B}^\dagger_{n,k,\ell}}
& \qw
& \qw
& \qw
& \qw
\\
\lstick{\ket{0}}
& \qw
& \qw
& \qw
& \qw
& \qw
& \qw
& \qw
& \ghost{\mathcal{B}_{n,k,\ell}}
& \ustick{_{b}} \qw
& \qw
& \ctrl{-3}
& \qw
& \ctrl{-2}
& \qw
& \ghost{*}
& \ctrl{-2}
& \qw
& \ctrl{-3}
& \qw
& \ghost{\mathcal{B}^\dagger_{n,k,\ell}}
& \qw
& \qw
& \qw
& \qw
\\
\lstick{x}
& {/} \qw
& \multigate{2}{\mathcal{A}_n}
& \qw
& \qswap
& \qw
& \ustick{_{m}} \qw
& \qw
& \ghost{\mathcal{B}_{n,k,\ell}}
& \qw
& \ctrl{-4}
& \qw
& \qw
& \qw
& \qw
& \qw
& \qw
& \qw
& \qw
& \ctrl{-4}
& \ghost{\mathcal{B}^\dagger_{n,k,\ell}}
& \qw
& \qswap
& \multigate{2}{\mathcal{A}_n^\dagger}
& \qw
\\
\lstick{y}
& {/} \qw
& \ghost{\mathcal{A}_n}
& \qw
& \qswap
& \qw
& \ustick{_{M}} \qw
& \qw
& \ghost{\mathcal{B}_{n,k,\ell}}
& \qw
& \qw
& \ctrl{-5}
& \qw
& \qw
& \qw
& \qw
& \qw
& \qw
& \ctrl{-5}
& \qw
& \ghost{\mathcal{B}^\dagger_{n,k,\ell}}
& \qw
& \qswap
& \ghost{\mathcal{A}_n^\dagger}
& \qw
\\
\lstick{\ket{0}}
& \qw
& \ghost{\mathcal{A}_n}
& \ustick{_{x>y}} \qw
& \ctrl{-2}
& \ctrl{2}
& \qw
& \qw
& \qw
& \qw
& \qw
& \qw
& \qw
& \qw
& \qw
& \qw
& \qw
& \qw
& \qw
& \qw
& \qw
& \ctrl{2}
& \ctrl{-2}
& \ghost{\mathcal{A}_n^\dagger}
& \qw
\\
\lstick{c_x}
& \qw
& \qw
& \qw
& \qw
& \qswap
& \ustick{_{c_m}} \qw
& \qw
& \qw
& \qw
& \qw
& \qw
& \ctrl{-6}
& \qw
& \qw
& \qw
& \qw
& \ctrl{-6}
& \qw
& \qw
& \qw
& \qswap
& \qw
& \qw
& \qw
\\
\lstick{c_y}
& \qw
& \qw
& \qw
& \qw
& \qswap
& \ustick{_{c_M}} \qw
& \qw
& \qw
& \qw
& \qw
& \qw
& \qw
& \ctrlo{-7}
& \qw
& \qw
& \ctrlo{-7}
& \qw
& \qw
& \qw
& \qw
& \qswap
& \qw
& \qw
& \qw
\\
&&&&&&&&&&&&&& \raisebox{1.2em}{Repeat for each $0 \le \ell \le 2^k$}
\gategroup{1}{9}{9}{21}{4.6em}{.}
}
\]
\caption{Circuit diagram of $\mathcal{C}_{f,n,k,1}$. The Toffoli gates with controls acting on the $x$ and $y$ registers are understood to be arrays of $n-k$ Toffoli gates between the corresponding qubits of the control and target registers. The gate marked $*$ adds a phase of $g_\ell^*(0^{n-k})$ if both qubits are $\ket{0}$ and otherwise does nothing. For convenience, several of the wires are labeled with the values they take on corresponding to variables in \Cref{alg:eval_f_fixed}.}
\label{fig:C_f_n_k_1}
\end{figure}
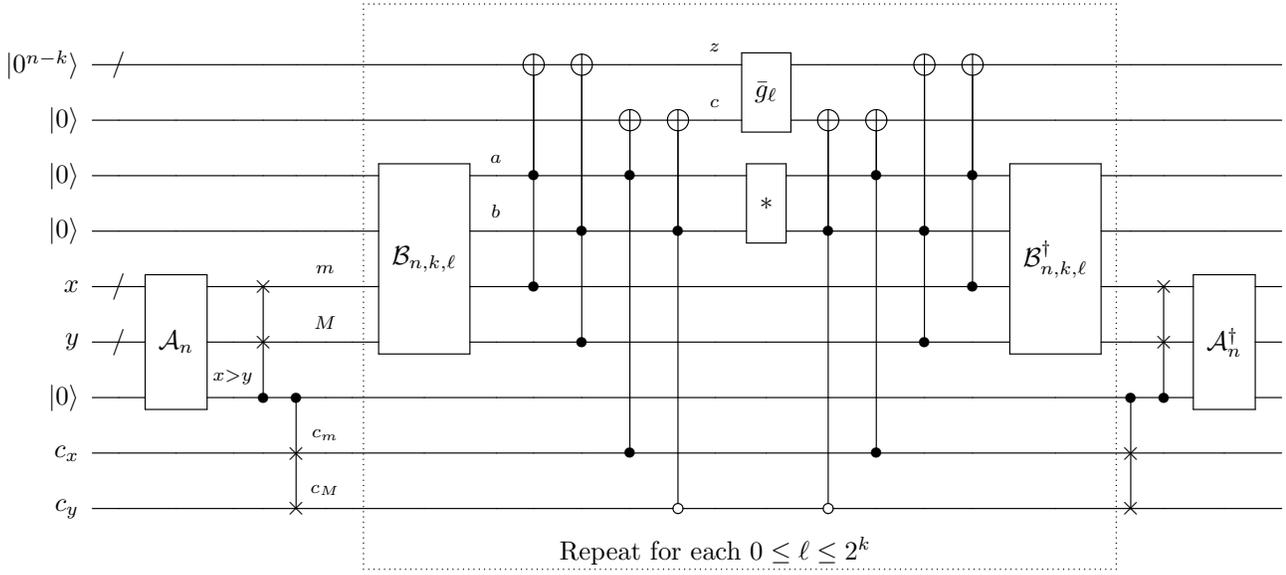

By \Cref{prop:multiplexed_rotation_naive}, for every $\ell$, $\bar{g}_\ell$ can be implemented using at most $2^{n-k}$ $\CNOT$ gates, because $\bar{g}_\ell$ is equivalent to an $(n-k,1)$-multiplexed $R_z$. Moreover, it is easy to see that $\mathcal{A}_n$ and $\mathcal{B}_{n,k,\ell}$ can be implemented using at most $O(n)$ $\CNOT$ gates each, because comparison of two $n$-bit integers can be performed by a classical circuit of at most $O(n)$ gates. As a consequence, we conclude that there exists a constant $d$ such that:
\begin{equation}
\label{eq:bound_s_nk1}
s_{n,k,1} \le \left(2^k + 1\right) \left(2^{n - k} + dn \right).
\end{equation}
This is certainly less than the bound in \eqref{eq:bound_s_nkt}, so this establishes the base case of the induction proof.

Now we proceed to the induction step on $t$. Suppose that for every $n > k \cdot (t - 1)$, we have a circuit $\mathcal{C}_{f,n,k,t-1}$ computing $\bar{f}^{\otimes 2^{t-1}}$ with $\CNOT$ count bounded by
\begin{equation}
\label{eq:bound_s_nkt-1}
s_{n,k,t-1} \le \left(2^k + 1\right)^{t-1} \left(2^{n - (t-1)k} + 2^{t-1} dn \right).
\end{equation}
To construct $\mathcal{C}_{f,n,k,t}$, we start by first taking $2^{t-1}$ copies of $\mathcal{C}_{f,n,k,1}$. Then, for each $0 \le \ell \le 2^k$, we replace each of the $2^{t-1}$ sub-circuits that compute $\bar{g}_\ell$ with $\mathcal{C}_{g_\ell,n-k,k,t-1}$. Then, the number of gates in $\mathcal{C}_{f,n,k,t}$ is bounded by:
\begin{align*}
s_{n,k,t} &\le \left(2^k + 1\right) \left(s_{n-k,k,t-1} + 2^{t-1} dn \right)\\
&\le \left(2^k + 1\right) \left(\left(2^k + 1\right)^{t-1} \left(2^{n - k - (t-1)k} + 2^{t-1} d(n-k) \right) + 2^{t-1} dn \right)\\
&\le \left(2^k + 1\right)^t \left(2^{n - tk} + 2^{t-1} dn \right) + \left(2^k + 1\right)2^{t-1} dn\\
&\le \left(2^k + 1\right)^t \left(2^{n - tk} + 2^t dn \right),
\end{align*}
where the first line substitutes \eqref{eq:bound_s_nkt-1} for the cost of the $\bar{g}_\ell$'s and otherwise uses the same bound as \eqref{eq:bound_s_nk1} for the non-$\bar{g}_\ell$ gates, and the second line applies the induction hypothesis \eqref{eq:bound_s_nkt-1}. This establishes the induction step, and thus \eqref{eq:bound_s_nkt} holds for every $n > k \cdot t$.

Choose $k = \lceil \log n \rceil$. Then:
\begin{align*}
s_{n,k,t} &\le \left(2^k + 1\right)^t \left(2^{n - tk} + 2^t dn \right)\\
&= 2^{kt} \left(1 + \frac{1}{2^k}\right)^t \left(2^{n - tk} + 2^t dn \right)\\
&\le 2^{kt} e^{t / 2^k} \left(2^{n - tk} + 2^t dn \right)\\
&\le 2^{kt} (1 + o(1)) \left(2^{n - tk} + 2^t dn \right)\\
&\le 2^{kt} (1 + o(1)) \left(2^{n - tk} + o\left(2^{n - tk}\right) \right)\\
&\le (1 + o(1))2^n,
\end{align*}
where we applied the exponential inequality in the third line, and used the assumption $t \le o(n / \log n)$ in the fourth and fifth lines. This proves the theorem.
\end{proof}

\Cref{thm:diagonal_mass_produce} straightforwardly generalizes to arbitrary multiplexed rotations and multiplexors with a single data qubit, as below.

\begin{corollary}
\label{cor:multiplexed_single_rotation}
Let $U$ be an $(n,1)$-multiplexed rotation, and let $r = 2^{o(n / \log n)}$. Then there exists a quantum circuit with at most $(1 + o(1))2^n$ $\CNOT$ gates to implement $U^{\otimes r}$.
\end{corollary}

\begin{proof}
The $R_z$ case follows by observing that $\bar{f}$ is exactly an $(n,1)$-multiplexed $R_z$ in \Cref{thm:diagonal_mass_produce}. This also extends to multiplexed $R_x$ and $R_y$, because multiplexed $R_x$, $R_y$, and $R_z$ are equivalent up to conjugation by single-qubit unitaries on the data qubit. That is, there exist single-qubit unitaries $U$ and $V$ such that:
\[
\hspace{2.5em}
\Qcircuit @C=1em @R=0.7em {
& {/} \qw
& \controlu \qw
& \qw
& \ccteq{1}
&&& {/} \qw
& \qw
& \controlu \qw
& \qw
& \qw
& \ccteq{1}
&&& {/} \qw
& \qw
& \controlu \qw
& \qw
& \qw
\\
& \qw
& \gate{R_z} \qwx
& \qw
& \ccteqg
&&& \qw
& \gate{U}
& \gate{R_x} \qwx
& \gate{U^\dagger}
& \qw
& \ccteqg
&&& \qw
& \gate{V}
& \gate{R_y} \qwx
& \gate{V^\dagger}
& \qw
}\\
\]
Hence, the $\CNOT$ count is identical for multiplexed $R_x$ and $R_y$ as well.
\end{proof}

\begin{corollary}
\label{cor:multiplexed_single_arbitrary}
Let $U$ be an $(n,1)$-multiplexor, and let $r = 2^{o(n / \log n)}$. Then there exists a quantum circuit with at most $(4 + o(1))2^n$ $\CNOT$ gates to implement $U^{\otimes r}$.
\end{corollary}

\begin{proof}
By \cite[Theorem 6]{SBM06-synthesis}, an arbitrary $(n,1)$-multiplexor may be implemented via a product of $4$ $(n,1)$-multiplexed rotations, as below.
\[
\hspace{2.5em}
\Qcircuit @C=1em @R=0.7em {
\lstick{\ket{0}}
& \qw
& \qw
& \qw
& \ccteq{2}
&& \lstick{\ket{0}}& \qw
& \qw
& \qw
& \qw
& \qw
& \qw
& \ccteq{2}
&& \lstick{\ket{0}}& \qw
& \qw
& \qw
& \qw
& \gate{R_z}
& \qw
\\
& {/} \qw
& \controlu \qw
& \qw
& \ccteqg
&&& {/} \qw
& \controlu \qw
& \controlu \qw
& \controlu \qw
& \gate{\Delta} \qw
& \qw
& \ccteqg
&&& {/} \qw
& \controlu \qw
& \controlu \qw
& \controlu \qw
& \controlu \qwx[-1] \qw
& \qw
\\
& \qw
& \gate{\phantom{U}} \qwx
& \qw
& \ccteqg
&&& \qw
& \gate{R_z} \qwx
& \gate{R_y} \qwx
& \gate{R_z} \qwx
& \qw
& \qw
& \ccteqg
&&& \qw
& \gate{R_z} \qwx
& \gate{R_y} \qwx
& \gate{R_z} \qwx
& \qw
& \qw
}
\]
Applying \Cref{cor:multiplexed_single_rotation} to each of the multiplexed rotations on the right side above completes the proof.
\end{proof}

\section{States and General Unitaries}

We now prove the main results of this work that generalize the mass production theorems above to state preparation and unitary compilation. The proofs proceed via the techniques of \cite{SBM06-synthesis}, by decomposing operators into multiplexors.

\thmstatemassproduce*

\begin{proof}
By \cite[Theorem 9]{SBM06-synthesis}, for any $n$-qubit quantum state $\ket{\psi}$, there exists an $(n-1)$-qubit state $\ket{\varphi}$ such that $\ket{\psi}$ has the following decomposition.
\[
\Qcircuit @C=1em @R=.7em
{
& {/} \qw
& \controlu \qwx[1] \qw
& \controlu \qwx[1] \qw
& \rstick{\ket{\varphi}} \qw
\\
& \qw
& \gate{R_z}
& \gate{R_y}
& \rstick{\ket{0}}\qw
\inputgroupv{1}{2}{1em}{1em}{\ket{\psi}}
}
\]
Applying this decomposition recursively, we conclude that $\ket{\psi}$ can be prepared by a circuit consisting of a pair of $(\ell,1)$-multiplexed rotations for each $1 \le \ell \le n-1$, and a pair of single-qubit gates.

Apply \Cref{cor:multiplexed_single_rotation} to the $(\ell,1)$-multiplexed rotations for each $\lceil n/2 \rceil \le \ell \le n-1$, and otherwise apply \Cref{prop:multiplexed_rotation_naive} $r$ times for each $1 \le \ell \le \lceil n/2 \rceil - 1$. Then the total number of $\CNOT$ gates to prepare $\ket{\psi}^{\otimes r}$ is upper bounded by
\begin{align*}
r \cdot \sum_{\ell = 1}^{\lceil n/2 \rceil - 1} 2^\ell + \sum_{\ell = \lceil n/2 \rceil}^{n-1} (1 + o(1)) 2^\ell
&\le r2^{\lceil n/2 \rceil} + (1 + o(1))2^n\\
&\le 2^{\lceil n/2 \rceil + o(n / \log n)} + (1 + o(1))2^n\\
&\le (1 + o(1))2^n\qedhere
\end{align*}
\end{proof}

\thmunitarymassproduce*

\begin{proof}
By \cite[Theorem 11]{SBM06-synthesis}, an arbitrary multiplexor can be expressed as below.
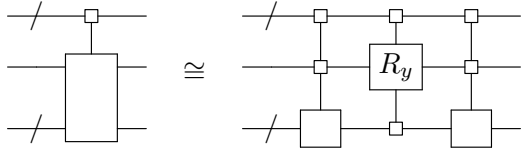
\[
\Qcircuit @C=1em @R=.7em
{
& {/} \qw
& \controlu \qwx[1] \qw
& \qw
& \ccteq{2}
&& {/} \qw
& \controlu \qwx[1] \qw
& \controlu \qwx[1] \qw
& \controlu \qwx[1] \qw
& \qw
\\
& \qw
& \multigate{1}{\phantom{U}}
& \qw
& \ccteqg
&& \qw
& \controlu \qwx[1] \qw
& \gate{R_y}
& \controlu \qwx[1] \qw
& \qw
\\
& {/} \qw
& \ghost{U}
& \qw
& \ccteqg
&& {/} \qw
& \gate{\phantom{U}}
& \controlu \qwx[-1] \qw
& \gate{\phantom{U}}
& \qw
\\
}
\]
This decomposition is also valid when the multiplexor on the left side of the equivalence has $0$ select bits. A recursive application of this decomposition implies that an arbitrary $n$-qubit unitary may be expressed as a product of $2^n - 1$ different $(n-1,1)$-multiplexors, of which $2^{n-1} - 1$ are multiplexed $R_y$ gates, and the remaining $2^{n-1}$ are arbitrary multiplexors. Applying \Cref{cor:multiplexed_single_rotation} and \Cref{cor:multiplexed_single_arbitrary} to these multiplexors gives the desired bound.
\end{proof}

\section{Conclusion and Outlook}
We have demonstrated that mass production phenomena are not unique to classical computation, and that they extend to quantum circuit complexity as well. As the message of this work is primarily conceptual in nature, we have not attempted to optimize every aspect of our results. Indeed, our mass production theorems could be extended further in a variety of ways; we outline a few such possibilities below.

If our results have any hope of being used in practice, then still more work needs to be done to optimize various constants. We suspect that the leading constant in \Cref{thm:unitary_mass_prod} could be brought down from $5/2$ to $1$ with a more clever decomposition into multiplexors. The factors hidden in the $o(1)$ could probably be optimized further as well, especially those related to the constant factor $d$ that appears in \Cref{thm:diagonal_mass_produce}. Indeed, we believe that much of the redundancy in computing and uncomputing $\mathcal{B}_{n,k,\ell}$ for each $0 \le \ell \le 2^k$ could be reduced by more careful accounting.

It is also worth attempting to optimize other parameters of practical relevance, such as constraints on the gate set, locality, depth, and ancilla qubit count. In principle, our proof should allow for some tradeoff between depth and ancilla count, because the $\bar{g}_\ell$s in \Cref{fig:C_f_n_k_1} can either be evaluated sequentially or in parallel. Another particularly interesting question is whether ancilla qubits are necessary at all to achieve quantum mass production.

We leave open the circuit complexity of quantum mass production in other parameter regimes. As \Cref{thm:state_mass_prod} and \Cref{thm:unitary_mass_prod} only apply when $r = 2^{o(n / \log n)}$, it is natural to ask what happens when $r$ is much larger. For Boolean functions, it is known that for any $n$-bit $f$, the ``asymptotic complexity'' of mass production $\lim_{r \to \infty} \frac{C(f^r)}{r}$ is bounded by $\poly(n)$ \cite{Pau76-disj,Alb89-simul}, where $C(f^r)$ denotes the Boolean circuit complexity of implementing $r$ copies of $f$. However, it is unclear whether the same approach would generalize to quantum circuits.

Lastly, we ask: are there any restricted examples of quantum circuits that exhibit a mass production phenomenon? What about Clifford circuits? We observe if one allows implementation by non-Clifford gates, then $n$ copies of an arbitrary Clifford operation can be implemented by a circuit with at most $O(n^\omega)$ gates, where $\omega$ is the exponent of matrix multiplication. By the ``canonical form theorem'' of Aaronson and Gottesman \cite{AG04-stabilizer}, every Clifford circuit can be expressed in the form H-C-P-C-P-C-H-P-C-P-C, where each letter corresponds to a layer of \underline{H}adamard, \underline{C}NOT, or \underline{p}hase gates. The Hadamard and phase layers contain at most $O(n)$ gates total, so it suffices to show how to implement $n$ copies of a $\CNOT$ circuit using $O(n^\omega)$ gates. For any $M \in \mathbb{F}_2^{n \times n}$, define $U_M$ as the unitary transformation that acts as $U_M \ket{x}\ket{y} = \ket{x}\ket{y \oplus Mx}$ on computational basis states. As every $\CNOT$ circuit implements an invertible linear transformation $\ket{x} \rightarrow \ket{Mx}$ for some $M \in \mathbb{F}_2^{n \times n}$, a $\CNOT$ circuit can be implemented using $U_M$ and $U_{M^{-1}}$ and $O(n)$ additional gates via:
\[
\ket{x}\ket{0^n} \xrightarrow{U_M} \ket{x}\ket{Mx} \xrightarrow{U_{M^{-1}}} \ket{0^n}\ket{Mx} \xrightarrow{\mathrm{SWAP}} \ket{Mx}\ket{0^n}.
\]
Then, as in \Cref{sec:intro}, we can mass produce $U_M$ and $U_{M^{-1}}$ using fast matrix multiplication.

\section*{Acknowledgments}
Part of this work was done while the author attended the 2022 Extended Reunion for the Quantum Wave in Computing at the Simons Institute for the Theory of Computing. We thank Alex Meiburg for helpful discussions.

\bibliographystyle{alphaurl}
\bibliography{../../../BibDesk/MainBibliography}

\end{document}